\documentclass[conference]{IEEEtran}

\usepackage[utf8]{inputenc} 
\usepackage{url,soul}
\usepackage{ifthen}
\usepackage{cite}
\usepackage[cmex10]{amsmath} 

\usepackage{color, tikz, graphicx, amssymb, mathtools, amsthm, comment, subcaption}

\DeclarePairedDelimiter\floor{\lfloor}{\rfloor}

\newcommand{\E}{\mathbb{E}}
\newtheorem{theorem}{Theorem}

\newtheorem{proposition}{Proposition}
\newtheorem{definition}{Definition}

\begin{document}
\title{Anonymity Mixes as (Partial) Assembly Queues:\\
Modeling and Analysis}

\author{%
  \IEEEauthorblockN{Mehmet Fatih Akta\c{s} and Emina Soljanin}
  \IEEEauthorblockA{Department of Electrical and Computer Engineering, Rutgers University \\
Email: \{mehmet.aktas, emina.soljanin\}@rutgers.edu}
}
\maketitle

\begin{abstract}
  Anonymity platforms route the traffic over a network of special routers that are known as mixes and implement various traffic disruption techniques to hide the communicating users' identities.
  Batch mixes in particular anonymize communicating peers by allowing message exchange to take place only after a sufficient number of messages (a batch) accumulate, thus introducing delay.
  We introduce a queueing model for batch mix and study its delay properties.
  Our analysis shows that delay of a batch mix grows quickly as the batch size gets close to the number of senders connected to the mix.
  We then propose a randomized batch mixing strategy and show that it achieves much better delay scaling in terms of the batch size. However, randomization is shown to reduce the anonymity preserving capabilities of the mix.
  We also observe that queueing models are particularly useful to study anonymity metrics that are more practically relevant such as the time-to-deanonymize metric.
\end{abstract}


\section{Introduction}
In numerous circumstances, more than just the content of a message has to be hidden from the adversary. 
Unlike covertness which aims to deny that any communication is taking place \cite{DeniableComm:KadheJB14, HidingInformationInNoise:BashGT15, CovertComm:Bloch16}, we consider the case where it is known that a group of peers communicate but it is desired to hide who is communicating with whom \cite{TorlikeAnonymityInWirelessNets:SafakaCA15}.
It is well known that identities of peers communicating over a network can be identified via rather simple network traffic analysis techniques \cite{NetsWithUserObservability:PfitzmannW85}.
Anonymity mixes were introduced by David Chaum in 1980's as a general framework for implementing anonymous message exchange \cite{Mixes:Chaum81}.
They are sophisticated network routers that pass messages such that no one (except the mix itself) can link an ingoing message to an outgoing message. Today, some form of a mix is often a part of anonymity preserving solutions (e.g., PetMail, Mixminion, Panoramix) or data transfer services (e.g., Onion routing, Freenet).

A mix typically collects messages and forwards them in batches according to a fixed deterministic rule or a randomized strategy (see e.g., \cite{BatchMix:KesdoganAP02, DummyMessageMix:BertholdL02, Mix:diazS03}). This allows hiding the origin of the outgoing messages, but also introduces delay in message transfer.
The incurred delay of the mixes is the most concerning cost of anonymity they provide. 
For instance anonymous web browsing platform ToR, which currently has more than 2 million users, does not implement sophisticated mixing to keep a low latency platform, even though it is shown to be vulnerable to deanonymization attacks based on network traffic analysis \cite{ChallengesDeployingTor:DingledineMS05, Tor:dingledineMS07,TorTraffAnalysis:ChakravartyBP14}.

Appropriate modeling of the mixes is crucial to study their delay vs.\ anonymity tradeoff.
Stability conditions and delay characteristics of a mix naturally depend on its system parameters which also determine its anonymity preserving capabilities.
In this paper, we propose and study two queueing models for batch mixes that are designed and used against passive adversarial attacks.
Note that, we do not consider active attacks that involve traffic injection into the network, which have also been shown to successfully deanonymize users on popular anonymity platforms \cite{SampledTrafficAnalysis:MurdochZ07, TorTraffAnalysis:ChakravartyBP14}.

We propose a mix model that implements the well known deterministic \textit{batch} mixing algorithm \cite{BatchMix:KesdoganAP02}.
We observe the close connection of the model to assembly queues, which was used to model and study the operational process of assembling multiple items into a product \cite{AssemblyQs:Harrison73}.
Using the proposed model, we find that batch mix provides a well defined anonymity guarantee that gets better with the batch size, on the other hand, its incurred delay grows quickly as the batch size gets close to the number of senders connected to the mix.

Our study of the batch mix led us to consider a new randomized mixing algorithm.
We show that the randomized model achieves better delay scaling in terms of the batch size compared to its deterministic counterpart. However, it can provably preserve anonymity \emph{only if} the adversary can not infer the state of the mix, and is in general vulnerable to anonymity attacks under low traffic.

There are many measures of anonymity and privacy (see \cite{Unlinkability:PfitzmannH05, InfoLeakage:IssaKW16}). We are here concerned with preserving \emph{unlinkability}, which is ensuring that no sender/receiver pair is exposed. Our study shows that delay of the mix can be reduced by sacrificing some anonymity, which would eventually lead to complete deanonymization of all the sender-receiver pairs.
However, 
message transfer sessions are of finite duration in practice, and minimum amount of time required for an attack to destroy anonymity is a concern regardless of the anonymity measure.
Previous papers that are concerned with the delay of anonymity schemes ignore the queueing dynamics within the mix (see e.g., \cite{DelayAnonymityTradeoffInMixNetworks:JavidbakhtV17, AnonymousNetworkingWithMinimumLatency:VenkitasubramaniamT08} and references therein).
We believe that queueing models are necessary for studying the time-to-deanonymize metric, and this paper is a first step towards understanding this metric.

This paper is organized as follows.
Sec.~\ref{sec:sec_batchmix_model} describes the batch mix model.
Sec.~\ref{sec:sec_batchmix_anonymity} presents the anonymity guarantee implemented by a batch mix.
Sec.~\ref{sec:sec_batchmix_stability_delay} shows a stability criterion for the batch mix and presents an approximate method for analyzing its incurred delay.
Sec.~\ref{sec:sec_samplingmix} introduces a randomized batch mixing strategy, and discusses its anonymity and delay properties. Sec.~\ref{sec:sec_conclusion} gives a summary and conclusions.

\section{A Batch Mix  and its Anonymity}
\subsection{Mix Model}
\label{sec:sec_batchmix_model}
A batch mix has $n$ senders connected, and buffers the messages incoming from each sender in a separate first-in first-out queue with an infinite buffer space. As soon as any $k \geq 2$ queues become non-empty, one message from each is dispatched (see Fig.~\ref{fig:fig_43_mix}).
The recipient sets of each sender are assumed to be disjoint and of at most size $m$.

\begin{figure}[t]
  \centering
  \includegraphics[width=0.4\textwidth, keepaspectratio=true]{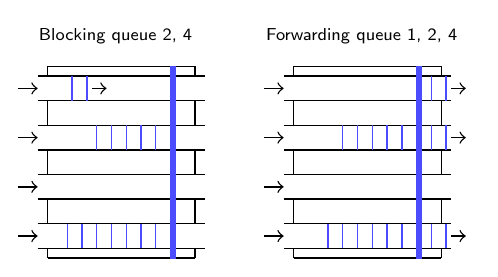}
  \caption{Illustration of a $(4, 3)$ batch mix. As long as there are less than three non-empty queues, messages are blocked (Left). As soon as a message arrival forms a group of three non-empty queues, one message from each is dispatched (Right).}
  \label{fig:fig_43_mix}
\end{figure}

Each sender is assumed to generate an independent Poisson
message traffic at rate $\lambda$.
Delay added by the mix is assumed to come only from the message queueing time. We ignore any message reception or transmission delay.

\begin{definition} 
  An \ul{$(n, k)$ batch mix} is a system of $n$ first-in first-out queues, each receiving messages from an independent Poisson process of rate $\lambda$.
  Messages are blocked as long as the mix has less than $k$ non-empty queues. As soon as $k$ queues become non-empty, one message from each is dispatched.
\end{definition}

\subsection{Attack Model and Anonymity}

\label{sec:sec_batchmix_anonymity}
 We assume that the adversary monitoring the traffic going in and out of the mix can observe 1) who the sender of each incoming message is and 2) who the recipient of each outgoing message is. Thus, if a message arrival triggers a message departure, the adversary can identify the sender-receiver pair. His goal is to identify the receivers of a particular sender, which we refer to as the \textit{target sender}.

Forwarding messages in batches of size $k$ prevents the adversary from immediately finding out the exact destination of an incoming message, as it can be any of the $k$ message recipients. However, the adversary can, over time, collect multiple size-$k$ receiver sets, each containing a potential recipient of the target sender. Intersecting such sets would eventually reveal the receivers linked to the target sender. We refer to such attacks as \textit{intersection attacks} \cite{IntersectionAttack:WrightAL02}.





We say that the mix preserves anonymity, when it ensures that no sender/receiver pair is exposed, that is, no sender and receiver can be linked.

\begin{theorem}[Anonymity under intersection attack]
  Consider a target sender connected to an $(n, k)$ batch mix that is under intersection attacks.
  When $k < n$, all $m$ receivers of the target can be identified if $m \leq \floor*{n/k}$.
  All $m$ receivers cannot be identified surely otherwise.
\label{thm_batchmix_anonymity}
\end{theorem}
\begin{proof}
  This theorem is a reformulation of \cite[Claim~1]{BatchMix:KesdoganAP02}.
  Let adversary wait and observe $m$ \emph{mutually disjoint} sets $R_1, \ldots, R_m$ of size $k$ that include the possible receivers of Alice. These $m$ sets can be disjoint only if $k m \leq n$. Adversary is thus sure that there is exactly one receiver of Alice in each observed recipient set $R_i$. Afterwards, adversary refines each of these sets by observing new recipient sets that intersect with only one of the prior sets. This means, a new recipient set $R$ is useful if $R \cap R_i \neq \emptyset$ and $R \cap R_j = \emptyset$ for all $j \neq i$, then $R_i$ can be refined to $R \cap R_i$. 
  Note that if $R$ intersects with multiple prior recipient sets, then refining all intersecting sets  may remove the actual receivers of Alice.
  The correct refinement process can be continued until each of the sets $R_1, \ldots, R_m$ contains only one receiver. Remaining $m$ receivers in the refined recipient sets are clearly the communication partners of Alice.
  
  As described above, intersection attacks will surely identify all receivers of a target only if adversary can observe $m$ \emph{disjoint} sets of size $k$. This is the only way for adversary to isolate each receiver of the target in a separate set so that a newly observed set can be intersected with one of these sets \emph{correctly}, that is, intersection will not surely end up removing the true receiver from the set.
  When $k m>n$, adversary can never observe $m$ disjoint sets of size $k$, hence can never surely identify all $m$ receivers of the target.
\end{proof}

\noindent
\section{Stability and Delay}
\label{sec:sec_batchmix_stability_delay}
A batch mix consists of $n$ FIFO queues, each buffering messages arriving from an i.i.d.\ Poisson process. A message arrival triggers a batch departure if it finds $k-1$ other non-empty queues in the mix, and the arriving message departs immediately with the batch. Therefore, there can be at most $k-1$ non-empty queues in the mix at any time. Since all the queues and the associated arrival processes are identical, system state can be represented as the Markov process $L(t) = (l_1(t), \dots, l_{k-1}(t))$ where $l_i(t)$ denotes the length of the $i$th longest queue in the system at time $t$.

An $(n, n)$ batch mix behaves as an assembly queue, found  to be unstable in \cite{AssemblyQs:Harrison73}.
Stability here refers to the existence of an invariant probability measure for the system state process.

\begin{theorem}
  An $(n, k)$ batch mix is stable if $k < n$.
\label{thm_nkmix_isstable}
\end{theorem}
\begin{proof}
  A Markov process is stable if and only if it is positive recurrent.
  Given that transition rates of $L(t)$ are neither too ``slow'' nor too ``fast'', its positive recurrence is implied by the positive recurrence of its embedded discrete chain $S_t$.
  We here use the Foster-Lyapunov criterion to show the positive recurrence of $S_t$ as interpreted from \cite[Thm.~2]{MarkovChainStability:FossK04}.
  
  For system state $s = (s_1, \dots, s_{k-1})$, let
  \[ W(s) \coloneqq s_{k-1}^{\log_2(n/n-1)}. \]
  Recall that $s_{k-1} = \min\{s_i, ~i=1, \dots, k-1\}$.
  
  Note that $\sup_s W(s) = \infty$ as required.
  One step drift for any state $s \in \{s, ~W(s) > 0\}$ is
  \[ \E\left[W(S_1) - W(S_0) \mid S_0=s\right] < 0. \]
  and we have
  \[ \sup_{\{s,\, W(s) \leq 0\}} \E[S_1 \mid S_0=s] < 1 < \infty. \]
  Thus $S_t$, hence $L(t)$ is positive recurrent.
\end{proof}

\begin{figure*}[t]
  \centering
  \begin{subfigure}[]{.32\textwidth}
    \centering
    \includegraphics[width=1\textwidth, keepaspectratio=true]{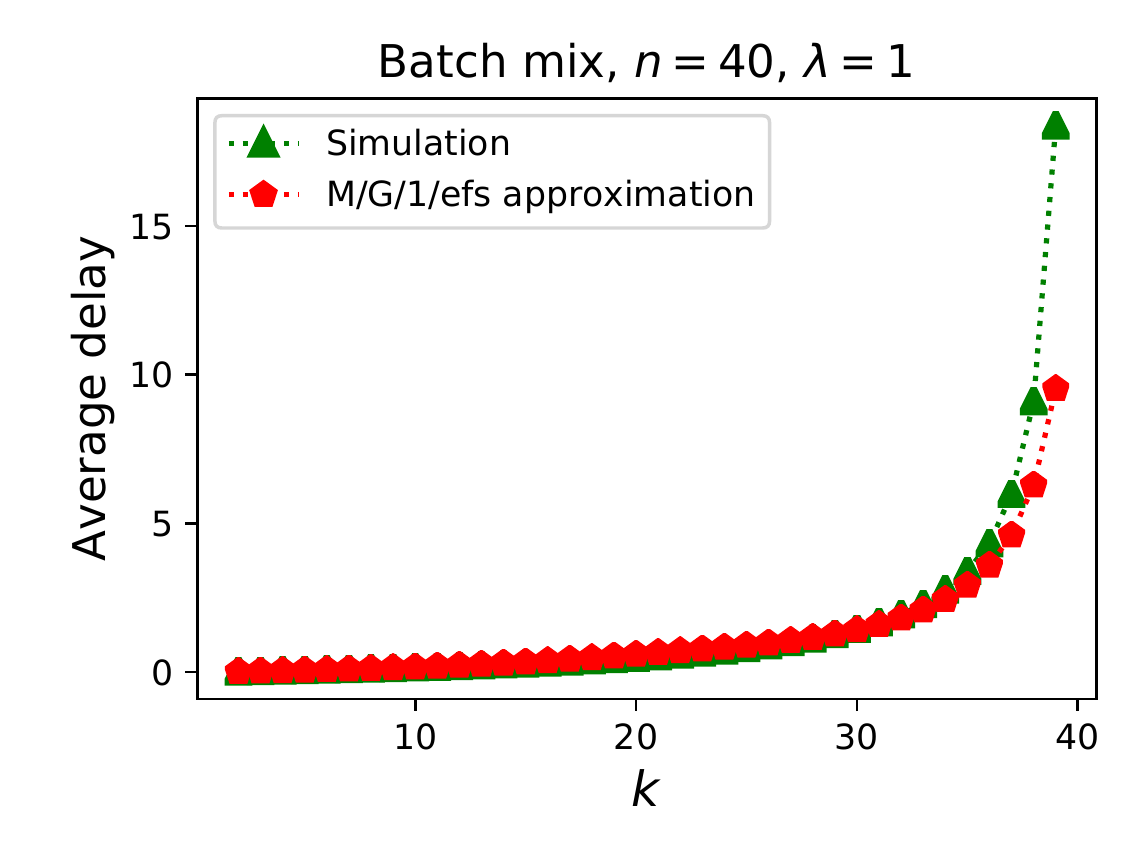}
  \end{subfigure}
  \begin{subfigure}[]{.32\textwidth}
    \centering
    \includegraphics[width=1\textwidth, keepaspectratio=true]{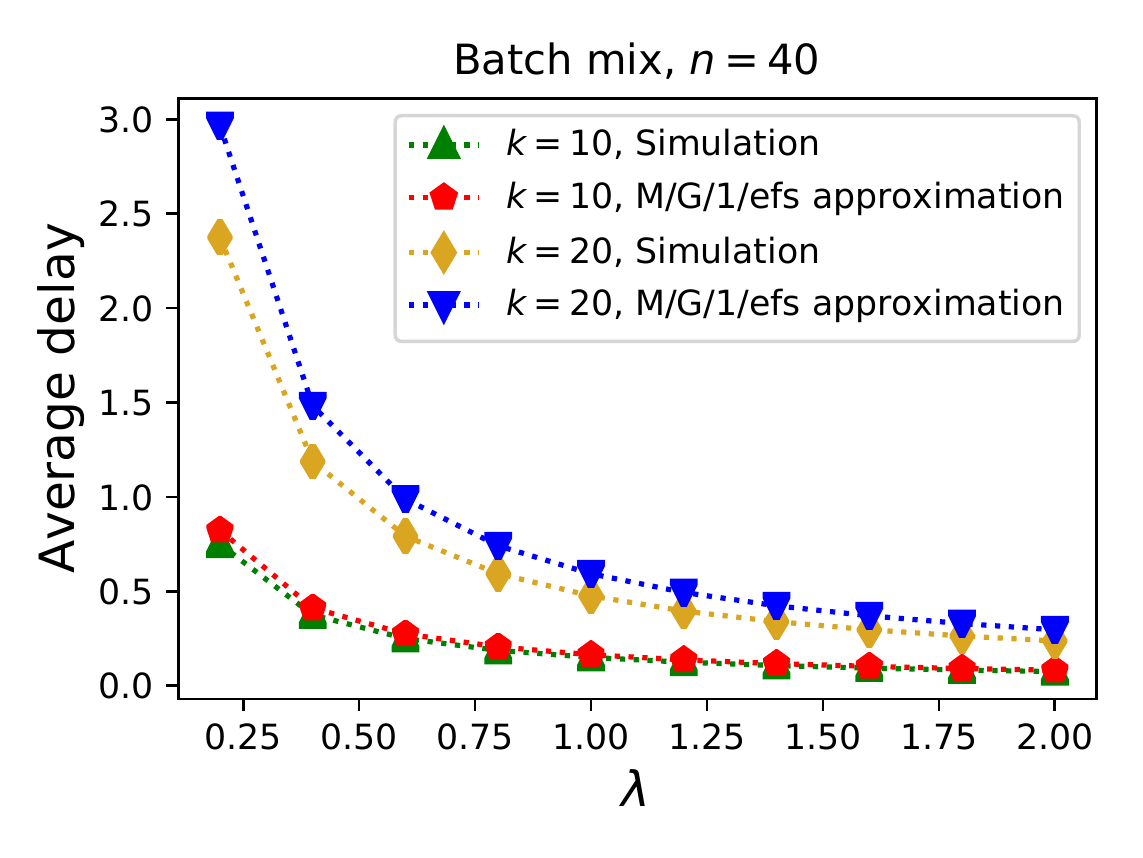}
  \end{subfigure}
  \begin{subfigure}[]{.32\textwidth}
    \centering
    \includegraphics[width=1\textwidth, keepaspectratio=true]{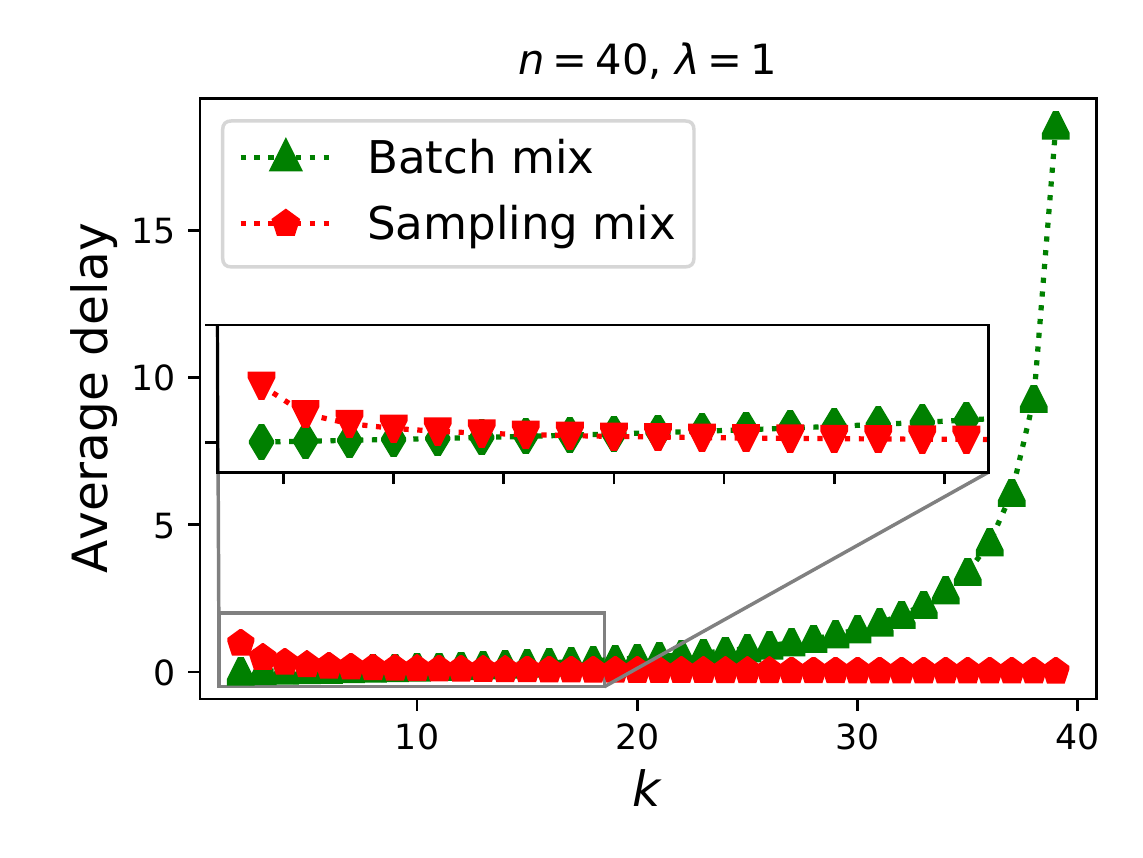}
  \end{subfigure}
  \caption{(Left, Middle) Average delay in a batch mix with $n=40$; Left: fixed $\lambda$ and varying $k$, Middle: fixed $k$ and varying $\lambda$. (Right) Comparison between the average delay in a batch mix and in a sampling mix with $p_a = 1/n$.}
  \label{fig:plot_ED_sim_vs_approx}
\end{figure*}

\noindent
There are three scenarios that a message can experience upon arrival to the mix:
1) If a message arrives to an empty queue and finds $k-1$ other non-empty queues in the mix, then it will immediately depart with no queueing.
2) If a message arrives to an empty queue and finds fewer than $k-1$ other non-empty queues in the mix, then it has to wait for a formation of $k$ non-empty queues (i.e., a batch).
3) If a message arrives to a non-empty queue, it has to first wait to the HoL (head of the line) in its queue, and then wait for the next batch formation.

In a tagged queue, batch formation delay experienced by a message is completely characterized by the number of non-empty queues $R$ (excluding the tagged queue) seen by the message once it moves to HoL. If $R < k-1$, message will be blocked until any $k-1-R$ of the $n-1-R$ empty queues receive at least a message.
Using the memoryless property of message inter-arrival times, batch formation delay is distributed as the $(k-1-R)$th order statistic of $n-1-R$ i.i.d.\ exponentials, which we denote as $X_{n-1-R:k-1-R}$\footnote{$X_{i:j} \coloneqq 0$ if $i < j$ or $j = 0$.}.
Overall, a message moving to HoL may find from $0$ up to $k-1$ other non-empty queues, hence there are $k$ possible different distributions for the batch formation delay.

When $k=2$, system state is just the longest queue length and defines a birth-death process. Exact analysis is formidable when $k > 2$ because of the infamous state space explosion problem. We first present the exact analysis for $k=2$, then present an approximate method for $k > 2$, which is similar to an approximation presented for assembly queues in \cite{AssemblyQs:LipperS86}.

\subsection{Exact analysis of $(n, 2)$-mix}
\label{subsec:subsec_n_2}
In $(n, 2)$-mix for $n > 2$, there can be at most one non-empty queue at any time, hence the system state is captured by the length of the longest queue $\mathcal{L}(t)$. It defines a single dimensional birth-death Markov process as shown in Fig.~\ref{fig:fig_n_2_markovprocess}.
\begin{figure}[t]
  \centering
  \includegraphics[width=0.35\textwidth, keepaspectratio=true]{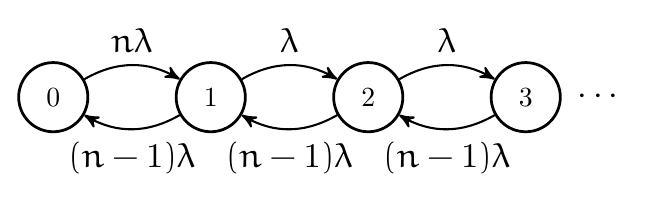}
  \caption{Markov process for $(n, 2)$-mix. State here denotes the length of the longest queue in the mix; $0$ length means the system is empty.}
  \label{fig:fig_n_2_markovprocess}
\end{figure}

Exact analysis in this case is straightforward. Let $p_l$ be the stationary probability for state $l$. From global balance equations we find
\[ p_0 = \frac{n-2}{2(n-1)} \quad\quad p_l = \frac{n(n-2)}{2(n-1)^{i+1}};~l=1,2,\ldots \]
Ergodicity implies that fraction of the time an arbitrary queue is non-empty (i.e., average load on the queue) is $\rho = (1-p_0)/n = 1/2(n-2)$. Larger $n$ gives higher frequency of emptiness at the servers, which is natural since queues empty out faster when the mix receives messages at a higher rate.

Using the stationary state probabilities, first two moments of the length of an arbitrary queue are given as
\[
\begin{split}
  \E[L] &= \frac{1}{n}\sum_{l=1}^{\infty} l\; p_l = \frac{1}{2(n-2)} \\
  \E[L^2] &= \frac{1}{n}\sum_{l=1}^{\infty} l^2\; p_l = \frac{n}{2(n-2)^2}
\end{split}
\]

We next derive some simple conclusions for the steady state delay experienced by an arriving message. Using PASTA \cite{Pasta:Wolff82}, an arbitrary message finds the system empty with probability $p_0$ and will have to wait for the first arrival to one of the other $n-1$ queues. Since arrivals are Poisson, waiting time distribution for the message is minimum of $n-1$ $\mathrm{Exp}(\lambda)$'s, that is $\mathrm{Exp}\left((n-1)\lambda\right)$. An arriving message may find its corresponding queue with $l$ messages with probability $p_l/n$. In this case, waiting time distribution for the message is sum of $l+1$ independent $\mathrm{Exp}\left((n-1)\lambda\right)$'s, which is $\mathrm{Erlang}\left(l+1, (n-1)\lambda\right)$. Finally, it may also find its corresponding queue empty with probability $p_l(n-1)/n$ if there is another queue with $l$ messages. Then the message will not be queued and will depart immediately upon arrival together with the first message in the busy queue. Using the law of total probability, distribution of waiting time $D$ for an arbitrary message is then given as
\begin{equation*}
\begin{split}
  \Pr\{D > w\} &= p_0\; \Pr\left\{\mathrm{Exp}\left((n-1)\lambda\right) > w\right\} \\
  &+ \frac{1}{n}\sum_{l=1}^{\infty} p_l\; \Pr\left\{\mathrm{Erlang}\left(l+1, (n-1)\lambda\right) > w\right\}.
\end{split}
\end{equation*}


\subsection{Approximate analysis of $(n, k>2)$-mix}
\label{subsec:subsec_n>2}
We here adopt the following approximating assumption;
a message \emph{upon moving to head of the line (HoL)} in its queue finds each other queue independently non-empty with probability $p$. Given that, and the fact that there can be at most $k-1$ non-empty queues at any time, the number of non-empty queues seen by a message moving to HoL is distributed as $R \sim B | \{B \leq k-1\}$ where $B \sim \mathrm{Binom}(n-1, p)$.
Given $R = r$, message will have to wait before getting dispatched for the first $k-1-r$ among all the $n-1-r$ empty queues to receive at least one arrival, that is, the message will experience a batch formation delay of $V|\{R=r\} \sim X_{n-1-r:k-1-r}$.
Then $V$ for an arbitrary message, which \emph{arrives to a non-empty queue in the first place}, is approximately distributed as
\begin{equation}
  \Pr\{V \leq v\} = \E_R\left[\Pr\{X_{n-1-R:k-1-R} \leq v\}\right],
  \label{eq:eq_MG1efs_Vcdf}
\end{equation}
where $\E_R$ denotes expectation with respect to $R$.

Batch formation delay for messages that arrive to an empty queue is differently distributed (than $V$ above) because they find each other queue non-empty with a different probability than messages that arrive to a non-empty queue.
Let a tagged queue be left empty by a departing message $m$. If the queue was left non-empty, the next message in line would have immediately moved to HoL. Then according to our earlier assumption, the number of non-empty queues left behind non-empty by $m$ is distributed as $R$.
Including the next arrival to the tagged queue, say message $m^+$, the next batch formation requires $k-R$ of the $n-R$ empty queues to receive at least an arrival.
Given that $m^+$ is among these first $k-R$ arrivals and messages are generated from i.i.d. streams, probability that $m^+$ is the $j$th among the $k-R$ arrivals is $1/(k-R)$.
Thus, batch formation delay experienced by a message arriving to an empty queue is approximately distributed as
\begin{equation}
  \Pr\{V_e \leq v\} = \E_R\left[\sum_{j=1}^{k-R} \frac{\Pr\{X_{n-R-j:k-R-j} \leq v\}}{k-R} \right].
\label{eq:eq_MG1efs_Vecdf}
\end{equation}

\begin{proposition}[Approximation by decoupling the queues]
  An $(n, k)$ batch mix approximately behaves to each sender as an $M/G/1/\text{efs}$ queue with regular service times distributed as \eqref{eq:eq_MG1efs_Vcdf} and exceptional first service times distributed as \eqref{eq:eq_MG1efs_Vecdf}.
\end{proposition}

Approximation requires estimating $p$, for which a natural estimate would be the average load $\rho$ for a queue, which is known for an $M/G/1/\text{efs}$ queue to be \cite{MG1ExceptionalFirstBusy:Welch64}
\[ \rho = \frac{\lambda \E[V_e]}{1 - \lambda\left(\E[V] - \E[V_e]\right)}. \]
Moments of $V$ and $V_e$ depend on $p$, hence on its estimate $\rho$. The equality above can be solved numerically to find a value for $\rho$.
Simulated and approximated values of delay are compared in Fig.~\ref{fig:plot_ED_sim_vs_approx} for a $(40, k)$ batch mix. Despite the strong independence assumptions employed in deriving the approximation, it compares well with the simulations for low values of $k$, which is the practically relevant case since the incurred delay must be kept below a threshold in practice.

\section{Sampling mix}
\label{sec:sec_samplingmix}
A sampling mix also implements an $(n, k)$ model; mix buffers the messages from each of the $n$ connected senders in a separate FIFO queue and each sender communicates with a disjoint set of at most $m$ receivers. However, buffered messages are forwarded differently compared to batch mix; as soon as a message arrives to the mix, $k$ queues are randomly selected and released. Releasing a queue allows it to forward a message if it is non-empty.
Specifically with probability $p_a$, the queue that receives the arrival is selected together with $k-1$ queues chosen uniformly at random from the remaining $n-1$ queues, or with probability $1-p_a$, the queue that receives the arrival is skipped and $k$ queues are selected uniformly at random from the rest of the queues.

\begin{theorem}
  \underline{Average load} of a queue in an $(n, k)$ sampling mix is given as 
  \[
  \rho = (1 - p_a)/(k - p_a)
  \]
  and the \underline{average delay} experienced by a message is given as 
  \[
  (1 - p_a)/\lambda(k-1).
  \]
\end{theorem}
\begin{proof}
  The length of a particular queue in the mix defines a birth-death process with a state space of non-negative integers and transition rates given for $i \geq 0$ as
  \[ \Pr\{i \to i+1\} = \lambda(1 - p_a), ~~ \Pr\{i+1 \to i\} = (n-1)\lambda p_o, \]
  where $p_o = p_a(k-1)/(n-1) + (1-p_a)k/(n-1)$.
  Stationary state probabilities are easily derived, using which average length of a queue is found, then Little's law is applied.
\end{proof}

As shown in Fig.~\ref{fig:plot_ED_sim_vs_approx}, average delay of a sampling mix scales much better with $k$ (i.e., decays as $1/(k-1)$) compared to a batch mix (i.e., grows exponentially in $k$ beyond a value).
However, a sampling mix cannot provide a well-defined anonymity guarantee while a batch mix can (see Thm.~\ref{thm_batchmix_anonymity}).


\begin{theorem}
  All receivers of a target sender connected to a sampling mix can be identified with intersection attacks by an adversary that can infer the state of the mix.
\label{thm_samplingmix_statefulattack}
\end{theorem}
\begin{proof}
  Queues in the mix will empty out infinite number of times under stability. Suppose that the adversary can detect whenever the mix becomes empty.
  Firstly, assume $p \neq 0$. Given that a message from a target finds the mix empty, the arriving message will be forwarded with probability $p$ or no message will depart. If the message is immediately forwarded, a receiver of the target will revealed.
  Number of times repeating this attack required to identify a receiver is geometric with $p$, hence attack will be almost surely successful in finite time.

  Secondly, assume $p \neq 0$. Given that a message from a target finds more than one non-empty queue in the mix, the following departure may include messages going only to the receivers of the non-target senders. This reveals which receiver does not belong to the recipient set of the target. Eventually adversary will be left with the correct set of receivers.
\end{proof}


Sampling mix will empty out more frequently and cannot often build a state complex enough to hide the origin of the outgoing messages when $k$ is larger and/or arrival rate $\lambda$ is lower, hence intersection attacks with state knowledge will resolve faster.
Moreover, even simple intersection attacks that do not require state knowledge can deanonymize a target connected to a sampling mix if $p_a$ is not chosen carefully.


\begin{theorem}
  All receivers of a target sender connected to an $(n, k)$ sampling mix with $p_a \neq 1/n$ can be identified with intersection attacks that do not require state knowledge.
\label{thm_samplingmix_statelessattack}
\end{theorem}
\begin{proof}
  Once a message arrives to a queue in steady state, probability of a departure from any other queue is $q = p_o\rho = (1-p_a)/(n-1)$.
  Suppose $m=1$ and $p_a > q$ ($p_a < q$). Adversary can record the message departures per arrival from a target sender.
  By the law of large numbers, the greatest (smallest) number of departures will almost surely be observed on the correct receiver in the limit.
  Same idea applies when each sender communicates with multiple receivers.
  Finally, $p_a = q$ if and only if $p_a = 1/n$.
  
  In other words, in order to preserve anonymity, it is necessary to maximize  the uncertainty within the steady state probabilities of message departures from the queues.
  R\'enyi entropy is commonly used for measuring uncertainty and uniform distribution is known to maximize it, which is achieved by setting $p_a = 1/n$.
\end{proof}

\section{Conclusion}
\label{sec:sec_conclusion}
We proposed a queueing model for batch anonymity mixes and showed that batch a mix with a deterministic message dispatching policy ensures that no sender-receiver pair is exposed (referred to as anonymity in this paper) under intersection attacks. On the other hand, its incurred delay on message transfer grows quickly as the batch size gets close to the number of connected senders.
We introduced a sampling mix model that implements a randomized message dispatching policy. Sampling mix permits an exact delay analysis, which allowed us to show that randomization allows cutting the tail of delay immensely, however, at the cost of giving up on the anonymity guarantee implemented by its deterministic counterpart.
We hope to next use our proposed queueing model to understand the performance of mixes in terms of the time-to-deanonymize metric vs.\ the incurred message transfer delay.



\section*{Acknowledgments}
This research is based upon work supported by the National Science Foundation 
under Grant No.~SaTC-1816404.

\bibliographystyle{IEEEtran}
\bibliography{references}

\end{document}